%% file: icra16c.tex
\documentclass[letterpaper, 10 pt, conference]{ieeeconf}  

\IEEEoverridecommandlockouts                              
\usepackage{hyperref}
\usepackage{ifthen}
\usepackage{flushend}
\usepackage{graphicx} 
\usepackage{amsmath} 
\usepackage{amsthm}

\newtheorem{prop}{Proposition}

\newtheorem{defi}{Definition}
\theoremstyle{definition}

\usepackage[vlined,linesnumbered,algoruled]{algorithm2e}
\usepackage{caption}
\usepackage{subcaption}
\usepackage{enumerate}
\usepackage{amssymb}
\let\labelindent\relax
\usepackage{enumitem}
\usepackage{comment}
\setlist{nolistsep}

\let\oldsqrt\sqrt
\def\sqrt{\mathpalette\DHLhksqrt}
\def\DHLhksqrt#1#2{%
\setbox0=\hbox{$#1\oldsqrt{#2\,}$}\dimen0=\ht0
\advance\dimen0-0.2\ht0
\setbox2=\hbox{\vrule height\ht0 depth -\dimen0}%
{\box0\lower0.4pt\box2}}

\usepackage{tikz}

\usetikzlibrary{calc,arrows,positioning,decorations.markings}
\usepackage{pgfplots}

\definecolor{bblue}{HTML}{4F81BD}
\definecolor{rred}{HTML}{C0504D}
\definecolor{ggreen}{HTML}{9BBB59}
\definecolor{ppurple}{HTML}{9F4C7C}

\renewcommand{\baselinestretch}{0.98}

\makeatletter
\newcounter{megaalgorithm}
\newenvironment{megaalgorithm}[1][htb]
  {
   \let\c@algocf\c@megaalgorithm
   \begin{algorithm}[#1]%
  }{\end{algorithm}}
\makeatother

\title{\LARGE \bf
Reduced Complexity Multi-Scale Path-Planning on Probabilistic Maps}

\author{Florian Hauer$^{1}$~~~ Panagiotis Tsiotras$^{2}$
\thanks{$^{1}$PhD candidate, School of Aerospace Engineering, Georgia Institute of Technology, Atlanta, GA 30332-0150{, Email:\rm fhauer3@gatech.edu}}%
\thanks{$^{2}$Professor, School of Aerospace Engineering and Institute for Robotics and Intelligent Machines, Georgia Institute of Technology, Atlanta, GA 30332-0150{, Email:\rm tsiotras@gatech.edu}}%
\thanks{Support for this work has been provided by ARO MURI award W911NF-11-1-0046 and ONR award N00014-13-1-0563}%
}

\begin{document}

\maketitle

\begin{abstract}
We present several modifications to the previously proposed MSPP algorithm that can speed-up its execution considerably.
The MSPP algorithm leverages a multiscale representation of the environment in $n$ dimensions.
The information of the environment is stored in a tree data structure representing a recursive dyadic partitioning of the search space. The information used by the algorithm is the probability that a node in the tree corresponds to an obstacle in the search space. Such trees are often created from mainstream perception algorithms, and correspond to quadtrees and octrees in two and three dimensions, respectively.
We first present a new method to compute the graph neighbors in order to reduce the complexity of each iteration, from $O(\vert V\vert^2)$ to $O(\vert V\vert \log \vert V\vert)$.
We then show how to delay expensive intermediate computations until we know that new information will be required, hence saving time by not operating on information that is never used during the search.
Finally, we present a way to remove the very expensive need to calculate a full multi-scale map with the use of sampling and derive an theoretical upperbound of the probability of failure as a function of the number of samples.
\end{abstract}

\section{Introduction}

The path-planning problem of an agent moving in a $n$-dimensional environment full of obstacles is considered.
It is assumed, without loss of generality, that the search space is contained within an $n$-dimensional hypercube.
Each dimension of this hypercube is divided by successive dyadic partitioning.
This partitioning generates a tree data structure that contains
all the environmental data, organized hierarchically with increasing resolution at each successive depth.
In the 2-D and the 3-D cases this multiresolution data structure is commonly known as a quadtree and an octree, respectively. The information stored in each node of the tree is the probability that the volume of the search space corresponding to this node is occupied by an
obstacle.
The problem is then to find an obstacle-free path between a start point and a goal point, or to report failure if no obstacle-free path exists.

Path-planning algorithms rely on perception algorithms to map the environment and localize the agent on the map. Commonly used data structures to represent perceived environments include multiresolution representations resulting from many common perception algorithms~\cite{endres3,octomap2010}.
The use of hierarchical multiresolution data structures is motivated by several observations:
First, information collected about the environment is not uniform, as each sensor has its own range, resolution and noise properties.
The information used by perception algorithms is then naturally multiscale; estimation and inference is often used
to extract the best information out of noisy measurements, leading to a probabilistic representation of the environment~\cite{octomap2010}.
Second, on-board computational resources might be limited, thus not allowing an agent to systematically use all perceived information. Furthermore, precise information about the environment far away from the agent might not be valid, or may even be irrelevant,
if the robot is far away from the obstacle.
A multiscale representation of the collected data allows to choose the resolution for each region of the space as needed.
For planning purposes, for instance, local information is typically important over the short-term (e.g., obstacle avoidance),
while far away information affects only long-term objectives such as reaching a goal or exploring the environment.


Several approaches have been used in the past to incorporate multiscale information during  planning.
Bottom-up approaches use the information at the finest resolution and then combine it in increasingly coarser  resolutions. 
Top-down approaches solve the path-planning problem at the coarser resolution and then progressively increase the resolution of the solution~\cite{kambhampati1986multiresolution,pai1998multiresolution}.
Using both approaches can lead to fast optimal algorithms, as shown in~\cite{lu2010beamlet}.
But the preprocessing of the data required by this approach is too expensive (in terms of time and memory) for online applications.
Another approach consists in using the information at different resolutions at the same time. This idea is explored in \cite{cowlagi2011hierarchical},
where areas near the current vehicle are represented accurately, while farther-away areas are coarsely-encoded
 by using a transformation on the wavelet coefficients. The approach is shown to be complete and very fast.

A similar approach is used to create a local map in~\cite{behnke2004local}, but quadtrees are used instead of wavelets and only local planning is done. Other algorithms have been developed using multi-resolution maps, but they are often applied to a \textit{given} non-uniform grid, without using the information at different resolution scales for the same region of the search space~\cite{ferguson2006using,petres2007path}.
More recently, the MSPP algorithm \cite{hauer2015multi} extended the work of \cite{cowlagi2011hierarchical} to $n$ dimensions in a reformulation using $2^n$ trees instead of wavelets.
The notion of $\varepsilon$-obstacles guaranteeing completeness for any value of the threshold $\varepsilon$ was also introduced in the same paper.

In this paper, we propose several modifications to the MSPP algorithm introduced in \cite{hauer2015multi} to accelerate computations and extend the range of potential applications.
Neighbor checking, a bottleneck operation of the MSPP algorithm, is reworked to reduce its complexity from $O(\vert V\vert^2)$ to $O(\vert V\vert \log \vert V\vert)$.
The order in which operations are executed is also modified in order to minimize the computations.
The full reduced graph is not computed, only its nodes are computed;
the edges are calculated on-the-fly, which saves computations; since most edges are never used by the search algorithm.
Finally, we introduce a way to work without a multi-scale map, but  instead we use a predicate to determine whether a point of the search space is an obstacle or not.
This modification allows us to remove the expensive map creation step.
It also saves memory, since the full multi-scale map does not need to be known in advance. The trade-off here is giving up completeness for the sake of increased runtime performance and memory reduction, as well as an extension of the range of potential application where accurate information about the environment is not known and is collected incrementally via sampling.


\section{Notation and Previous Work} \label{sec:notations}

In this section, we briefly introduce the notation used in the rest of the paper.  We refer the reader to \cite{hauer2015multi} for more details on the original MSPP algorithm and the notation used throughout the rest of the paper.

\subsection{Multiresolution World Representation}

The environment $\mathcal{W} \subset \mathbb{R}^d$ is assumed to be a $d$-dimensional grid world.
Without loss of generality, we assume that each elementary cell of the grid world is a unit hypercube and there exists an integer $\ell>0$ such that the world is contained within a hypercube of side length $2^{\ell}$.
The world $\mathcal{W}$ is encoded  as a tree $\mathcal{T}=(\mathcal{N},\mathcal{R})$ representing the multi-scale information, with $\mathcal{N}$ being the set of nodes and $\mathcal{R}$ being the set of edges describing their relations.
Nodes of $\mathcal{T}$ are represented by two indices, $k$ and $p$, corresponding to the depth of the node $n_{k,p}$ in the tree and the position of the center of the node in the search space.
A node $n_{k,p}$ represents a hypercube in the search space $\mathcal{W}$ centered at $p$ and of size $2^{k}$, and is denoted by $H(n_{k,p})$.
The \textit{children} of the node $n_{k,p}$ are denoted by $n_{k-1,q_i}, i\in[1,2^d]$ where $q_i=p+2^{k- 2}e_i$ and where $e_i$ is each of the $2^d$ ($d$-dimensional) vectors generated by $[\pm 1,\pm 1, \dots, \pm 1]$.
A node is called a \textit{leaf} of $\mathcal{T}$ if it has no children.

The information $V(n_{k,p})$ contained in each node $n_{k,p}\in \mathcal{N}$ is the probability of the existence of an obstacle in the space represented by the node, computed by
\begin{equation}
V(n_{k,p})=\frac{\text{Volume of obstacles in $H(n_{k,p})$}}{\text{Volume of $H(n_{k,p})$}}.
\end{equation}

\subsection{The Path-Planning Problem}

Two nodes in $\mathcal{T}$ are \textit{neighbors}  if their corresponding hypercubes share a hyperface, specifically, their
intersection is a hypercube of dimension $d-1$. 
A necessary and sufficient condition for two  nodes $n_{k_1,p_1}$ and $n_{k_2,p_2}$ to be neighbors is that both of the following two conditions are satisfied:

\begin{itemize}
\item The expression
\begin{equation}
\label{nc1}
\|p_1-p_2\|_\infty = 2^{k_1-1}+2^{k_2-1}
\end{equation}
holds,
\item
There exists a unique $i\in[1,d]$, such that
\begin{equation}
\label{nc2}
\vert(p_1-p_2)_i\vert=2^{k_1-1}+2^{k_2-1},
\end{equation}
where $(p_1-p_2)_i$ is the $i^{th}$ component of the vector.
\end{itemize}

In the case of a 2-D or a 3-D environment with a uniform grid, conditions (\ref{nc1}) and (\ref{nc2}) imply 4-connectivity or 6-connectivity, respectively.

We define a \textit{path} $\pi=(n_{k_1,p_1},n_{k_2,p_2},\dots,n_{k_N,p_N})$ in $\mathcal{T}$ to be a sequence of nodes $n_{k_i,p_i} \in \mathcal{N}$, each at corresponding  position $p_i$ and depth $\ell-k_i$, such that two consecutive nodes of the sequence are neighbors.
A path is called a \textit{finest information path} (FIP) if all its nodes are leafs of $\mathcal{T}$.
Leaf nodes represent the best resolution, and hence the finer information contained in the tree $\mathcal{T}$.

Due to noisy measurements, and to prevent overconfidence and numerical issues, perception algorithms will not allow $V(n_{k,p})$ to reach 100\%. We therefore introduce the notion of $\varepsilon$-obstacles.
Specifically, given $\varepsilon \in [0,1) $, a node $n_{k,p}\in \mathcal{N}$ is an \textit{$\varepsilon$-obstacle} if
\begin{equation}
\label{obstacle_def}
V(n_{k,p})\geq 1-2^{-d k}\varepsilon.
\end{equation}
A path $\pi$ is \textit{$\varepsilon$-feasible} if none of its nodes are $\varepsilon$-obstacles.

Given the representation of $\mathcal{W}$ encoded in the tree $\mathcal{T}$,
the problem is to find an $\varepsilon$-feasible FIP between two nodes in the tree, $n_{\rm{start}}$, representing the starting node, and $n_{\rm{goal}}$, representing the goal node, and to report failure if no such path exists.

\subsection{The MSPP Algorithm}

The MSPP algorithm is a backtracking algorithm that iteratively builds a solution from the starting node $n_{\rm{start}}$ until the goal node $n_{\rm{goal}}$ is reached.
At each iteration $i$, a local representation $\mathcal{G}_i$ of the environment, called the \textit{reduced graph}, is computed and the best path to the goal on this graph is used as a heuristic to decide which direction to follow.
The current candidate solution built by the algorithm at iteration $i$ is denoted by $\pi_{\rm{start}}^i$.
Therefore, $\pi_{\rm{start}}^i$ is an $\varepsilon$-feasible FIP from $n_{\rm{start}}$ to $n_{k_i,p_i}$, the node reached by the algorithm at iteration $i$.
The best path from $n_{k_i,p_i}$ to the goal on the reduced graph $\mathcal{G}_i$ is denoted by $\pi_i^{\rm{goal}}$.
The reduced graph
$\mathcal{G}_i$ is computed by first identifying its vertices (corresponding to nodes of $\mathcal{T}$) via a top-down exploration of $\mathcal{T}$.
The first element of $\pi^{\rm{goal}}_i$ is used to build the global solution $\pi_{\rm{start}}^i$. If the path $\pi^{\rm{goal}}_i$ does not exist, the algorithm backtracks.

The main lines of the MSPP algorithm are shown in Algorithm~\ref{algo}. Refer to \cite{hauer2015multi} for full details of the MSPP algorithm.

\begin{algorithm}[ht]
 \KwData{Tree $\mathcal{T}$, Start node $n_{{\rm start}}$, Goal node $n_{{\rm goal}}$}
 \KwResult{$\varepsilon$-feasible FIP from $n_{{\rm start}}$ to $n_{{\rm goal}}$ or failure}
 $i\gets 0,n_{k_i,p_i}\gets n_{{\rm start}}$,$\pi_{{\rm start}}^0\gets [n_{k_i,p_i}]$\;
 \While{Goal not found AND no failure}{
  $(\tilde{\mathcal{G}}_i,v_{{\rm start},i},v_{{\rm goal},i})\gets$ReducedGraph($T$,$n_{k_i,p_i}$)\;
  $\pi_i^{{\rm goal}} \gets $ShortestPath($\tilde{\mathcal{G}}_i,v_{{\rm start},i},v_{\rm{goal},i}$)\;
  \eIf{exists($\pi_i^{{\rm goal}}$)}{
   $n_{k_{i+1},p_{i+1}}\gets $firstElement$(\pi_i^{{\rm goal}})$\;
   $\pi_{\rm{start}}^{i+1} \gets [\pi_{{\rm start}}^i \quad n_{k_{i+1},p_{i+1}}]$\;
   }{
   	backtrack\;
  }
  $i\gets i+1$\;
 }
 \caption{The MSPP Algorithm - Simplified for clarity}
 \label{algo}
\end{algorithm}

\section{Fast Neighbor Computation - MSPP-FN}

In the original MSPP algorithm~\cite{hauer2015multi}, neighbors for the reduced graph are computed by testing whether every pair of vertices satisfies the neighborhood properties.
As shown in \cite{hauer2015multi}, this step is the bottleneck during each iteration, having complexity $O(\vert V\vert ^2)$,
where $\vert V\vert$ is the number of vertices.
A new way to compute neighbors by reducing the complexity from $O(\vert V\vert^2)$ to $O(\vert V\vert\log \vert V\vert)$ proceeds as follows; for each vertex, neighbor candidates are generated, and the existence of these candidate vertices of $\mathcal{G}_i$ is verified. The details are given next. MSPP-FN will be used in the rest of the paper to refer to the MSPP algorithm using the fast neighbor computation.

\subsection{Tree Data Structure For Vertices}

In order to do fast searches over the vertices of the reduced graph, we keep them in a tree structure.
Let $\mathcal{T}_i$ define this tree structure.
As the original tree $\mathcal{T}$ is traversed to select nodes for $\mathcal{G}_i$, $\mathcal{T}_i$ is constructed by copying every element of $\mathcal{T}$ traversed by the selection process, except for $\varepsilon$-obstacles.
$\mathcal{T}_i$ is then a tree with the same structure as $\mathcal{T}$, but its branches are shorter.
In other words, $\mathcal{T}_i$ is a pruned version of $\mathcal{T}$ whose leave nodes are the vertices of $\mathcal{G}_i$.

Note that, for implementation, $\mathcal{T}_i$ does not change significantly between two consecutive iterations, so it is computationally cheaper to modify $\mathcal{T}_{i-1}$ than to create a new data structure at each iteration.
Memory allocation is the most expensive operation when creating new nodes.
Modifying $\mathcal{T}_{i-1}$ allows to only have to allocate memory for nodes of $\mathcal{T}_{i}$ that did not exist in $\mathcal{T}_{i-1}$.
The copy process is then modified to add nodes only if they do not already exist, and remove excessive nodes when reaching a node corresponding to a vertex of $\mathcal{G}_i$.
The pseudo-code is given in Function~\ref{algo:grgv2}.
The vertex list is also removed since the information is already contained in $\mathcal{T}_i$.
The function {\tt GetRGFastNeighbor} is called with the root of $\mathcal{T}$, the root of $\mathcal{T}_i$ (created during initialization) and the current node $n_{k_i,p_i}$.

\begin{megaalgorithm}[ht]
 \KwData{Node $n_{k,p}$ (in $\mathcal{T}$), Node $t_{k,p}$ (in $\mathcal{T}_i$), Current node $n_{k_i,p_i}$}
 \eIf{($\|p-p_i\|_2-\frac{\sqrt{d}}{2}2^{k_i} \geq \alpha 2^{k}$ OR isLeaf($n_{k,p}$)) AND doesNotContainPath($n_{k,p}$)}{
 	\eIf{$n_{k,p}$ is not a $\varepsilon$-obstacle}{
 		Remove all descendants of $t_{k,p}$\;
 	}{
		Remove $t_{k,p}$ and its descendants\; 	
 	}
 }{
 	\ForEach{$(m,q)$ index of children of $(k,p)$}{
 		\If{$t_{m,q}$ does not exist}{Create $t_{m,q}$ child of $t_{k,p}$\;}
 		\tt{GetRGFastNeighbor}($n_{m,q}$,$t_{m,q}$,$n_{k_i,p_i}$)\;
 	}
 }
\caption{\tt{GetRGFastNeighbor}()}
\label{algo:grgv2}
\end{megaalgorithm}


\subsection{Same Size Neighbors}

Generating neighbors is easy when the nodes have the same size, so we will consider this case first. Given a node $n_{k,p}$, we want to find all its neighbors
having the same size that correspond to vertices of $\mathcal{G}_i$.
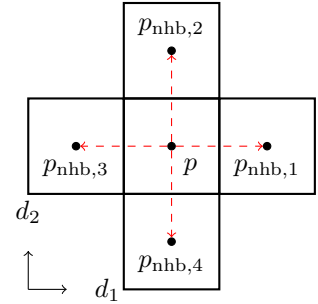
\begin{figure}[ht]
\centering
	\input{same_size.tex}
	\caption{Generating neighbors for the construction of $\mathcal{G}_i$. Same size neighbors case: $H(n_{k,p})$ is the center square and the $p_{\mathrm{nhb},i}$ are the generated position candidates for the neighbors.}
	\label{same_size}
\end{figure}
Same size implies the same depth in $\mathcal{T}$, so every neighbor will have the same depth index $k$.
Also, the neighbor conditions and the fact that the nodes are centered on a grid, imply that only one dimension of the position vector can be changed at a time, that is, the neighbors' positions $p_{\mathrm{nhb},i}$ can only be
\[ p_{\mathrm{nhb},i}=p+2^k b_i, 1\leq i \leq 2d\]
with
\[ b_i=\begin{cases}
                        d_i &\text{ if }i\leq d, \\
                        -d_i &\text{ otherwise,}
                    \end{cases} \]
where $d_i$ is the $i^{th}$ vector of the standard basis of $\mathbb{R}^d$.
If $p_{\mathrm{nhb},i}$ is within the bounds of the search space, $n_{k,p_{\mathrm{nhb},i}}$ is in $\mathcal{T}_i$ and $n_{k,p_{\mathrm{nhb},i}}$ is a leaf of $\mathcal{T}_i$, then $n_{k,p_{\mathrm{nhb},i}}$ is a valid neighbor of $n_{k,p}$.
Figure~\ref{same_size} shows in the center the hypercube corresponding to $n_{k,p}$ and the neighbor candidates around it. The red arrows represent the vectors $2^k b_i$.

Searching the tree $\mathcal{T}_i$ can be done on average in $O(\log \vert V \vert)$, and the number of candidates to check is $2d$.

\subsection{Larger Neighbors}

Consider now the case of finding the larger neighbors of $n_{k,p}$. The previous result can still be used, but it will generate points inside larger neighbors instead of their positions.
\begin{figure}[ht]
\centering
	\input{larger.tex}
	\caption{Generating neighbors for the construction of $\mathcal{G}_i$. Larger neighbors case: $H(n_{k,p})$ is the smaller square  around $p$ and $p_{\mathrm{nhb},1}$ is the first generated position candidate for the neighbors. The dashed squares represents the hypercubes corresponding to nodes visited during the search for $p_{\mathrm{nhb},1}$ in the tree.}
	\label{larger}
\end{figure}
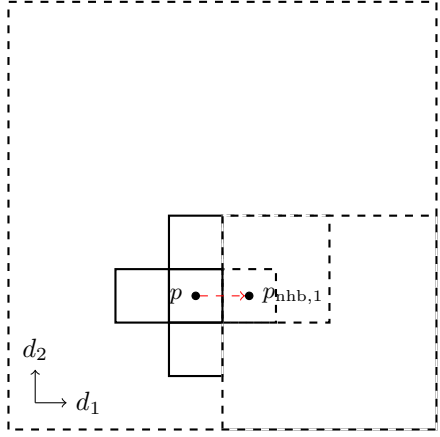
The search through the tree works as follows.
It starts with the root of the tree, representing the entire environment, as the current node. As long as the current node has children (recall that our data structure assumes that they either all exist or none of them exists), the child whose hybercube contains the searched point $p_{\mathrm{nhb},i}$ is selected as the current node. That is, at each step, the search process selects the node at the next level of resolution whose hypercube contains $p_{\mathrm{nhb},i}$. The search stops if either the current node is at $p_{\mathrm{nhb},i}$ or the current node does not have children. At the end of the process, the current node is a neighbor of $n_{k,p}$ and if it is a leaf, then it is also a vertex of $\mathcal{G}_i$.
A larger node could contain $p$ and then not be a neighbor, but since $n_{k,p}$ exists, that node would have children and the search process will never stop in such a situation.

Figure~\ref{larger} shows, in dashed lines, the last four nodes that would be explored while searching for $p_{\rm{nhb},1}$. If $n_{k,p_{\rm{nhb},1}}$ does not exist, the algorithm will stop at one of its ancestors, which will be a neighbor of $n_{k,p}$. Note that the search cannot stop at the largest ancestor shown, since it contains $n_{k,p}$, so all the children exist.

\subsection{Smaller Neighbors}

The last case to consider is when there are smaller neighbors of $n_{k,p}$. The search for $p_{\mathrm{nhb},i}$ in $\mathcal{T}_i$ will return a node that is not a leaf.
\begin{figure}[ht]
\centering
	\input{smaller.tex}
	\caption{Generating neighbors for the construction of $\mathcal{G}_i$.  Smaller neighbors case: $H(n_{k,p})$ is the left square and $p_{\mathrm{nhb},1}$ is the first generated position candidate for the neighbors. The larger dashed square is $H(n_{k,p_{\mathrm{nhb},1}})$ and the blue squares correspond to the descendant of $n_{k,p_{\mathrm{nhb},1}}$ that are neighbors of $n_{k,p}$.}
	\label{smaller}
\end{figure}
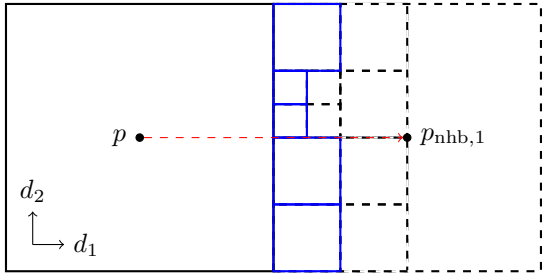
In this case, the exploration of children of $p_{\mathrm{nhb},i}$ can lead to the neighbors. Note that $p_{\mathrm{nhb},i}$ was generated by moving in the direction $b_i$, but since the neighbors are smaller, the move was too large, and hence neighbors of $n_{k,p}$ are leaf nodes, descendant of $p_{\rm{nhb},i}$ in the direction $-b_i$.

Figure~\ref{smaller} shows what happens for smaller neighbors. The larger dashed square is the hypercube corresponding to the candidate neighbor $p_{\rm{nhb},1}$, but that node is not a leaf, so it is not a vertex of $\mathcal{G}_i$. Exploring its children (until leaf nodes) in the direction $-b_1=-d_1$, will lead to all its descendants that are neighbor with $n_{k,p}$, and in $\mathcal{G}_i$, since they will be leaf nodes. The neighbors are drawn in blue in Figure~\ref{smaller}.

\subsection{Computing All Neighbors in \texorpdfstring{$\mathcal{G}_i$}{Gi}} \label{sec:allneigh}

When looking for all neighbors, nodes can be treated in any order, in particular, from smallest to largest. For the smallest nodes, all neighbors are larger. If all smaller nodes have been treated before, for a given node $n_{k,p}$, the smaller neighbors will already have been found and the only information missing is the larger nodes.
All neighbor pairs can then be found by looking for larger neighbors for each node ordered from the smallest to the largest.
Finding larger and same size neighbors is done in  $O(\log \vert V \vert)$ for each of the $\vert V \vert$ nodes, so finding all neighboring pairs in $\mathcal{G}_i$ is then be done in $O(\vert V \vert \log \vert V \vert)$.

\subsection{Computing All Neighbors of a Given Node \texorpdfstring{$n_{k,p}$}{nkp}}

Finding all neighbors of a given node $n_{k,p}$ can be done using the pseudo-code in Function~\ref{algo:findNeighbor}.
For each direction $b_i$, we compute the candidate neighbor position $p_{\mathrm{nhb},i}$ and search in $\mathcal{T}_i$ for the corresponding node.
If the node is a leaf, it means that a larger or same size neighbor has been found;
otherwise, the leaf descendants, in the direction $-b_i$, of the node found are smaller size neighbors.

\begin{megaalgorithm}[ht]
 \KwData{Node $n_{k,p}$}
 $neighbors$=$\emptyset$\;
 \ForEach{$i$ in $[1,2d]$}{
 		$n_{m,q}$=find($\mathcal{T}_i$,$p_{\rm{nhb},i}$)\;
 		\eIf{isLeaf($n_{m,q}$)}
 		{$neighbors$=$neighbors \cup n_{m,q}$\;}
 		{{\tt addLeafInDir}($n_{m,q}$, -$b_i$, $neighbors$)\;}
 	}
 	\Return $neighbors$\;

 \caption{\tt findNeighbors()}
\label{algo:findNeighbor}
\end{megaalgorithm}

\begin{megaalgorithm}[ht]
 \KwData{Node $n_{k,p}$, Direction $b$, List $neighbors$}
 \ForEach{$i$ in $[1,2^d]$}{
 	\If{$b^Te_i > 0$}
 	{$n$=child($n_{k,p}$,$i$)\;
 	\eIf{isLeaf($n$)}
 	{$neighbors$=$neighbors \cup n$\;}
 	{{ \tt addLeafInDir}($n$, $b$, $neighbors$)\;}}
 	}

 \caption{\tt addLeafInDir()}
\label{algo:addLeafInDir}
\end{megaalgorithm}

\section{Multi-Scale Path Planning Without Full Information Map - MSPP-S}

Although in 2D or 3D geometric workspaces, the multi-scale map is often the result of perception algorithms, this is not always the case.
When the search space is the configuration space and it is different from the geometric workspace, computing the multi-scale map might be very expensive, as it requires to analyze every single cell of the map.
Furthermore, in some cases, we may only have access to a predicate of whether a point of the search space is an obstacle or not.
A robotic arm, for example, is usually parameterized by the position of each joint; given a configuration, the spatial position of each link can be computed, and self-collision or collision with obstacles is checked in the geometric workspace.
It is assumed in this section that we have such a predicate, say $isObstacle(s)$, that informs us if a point $s$ of the search space is an obstacle.

In the proposed approach, sampling is used to estimate the obstacle probabilities of the nodes in $\mathcal{G}_i$.
Since we are using an estimate instead of the exact node probability values, completeness of the algorithm is not ensured.
Note, however, that if a large enough number of samples is drawn, the estimated probabilities will be close to their actual values,
and loss of completeness is very unlikely.

Similarly to the original MSPP algorithm, the proposed algorithm, MSPP-S (as MSPP with sampling), decomposes the space using a grid, which
has fine resolution near the current position, and the resolution becomes increasingly coarser farther away.
An empty tree data structure $\mathcal{T}_i$ is created to represent that grid.
It is empty in the sense that it does not have any information about the obstacles, it is a pure geometric partition of the search space.
For each node $n_{k,p}$ of the partition, the predicate can be used for a given number $N_{\mathrm{samples}}$ of random points drawn in the search space corresponding to the node.
An estimate of the probability of obstacles can then be calculated from those results and used to fill up the tree $\mathcal{T}_i$.
The value of the node is approximated by
\[
 \hat{V}(n_{k,p})=\frac{\text{Number of obstacles sampled}}{N_{\rm{samples}}}.
 \]
Similarly to the data structure used for neighbor checking, it is less costly to modify the data structure from the previous iteration than to recreate a new data-structure at each iteration.
Moreover, in that case, some information will already exist in the data structure, and sampling only needs to be done for the newly added nodes. 

Note that the data structure created is similar to the one created for neighbor checking, hence the same notation $\mathcal{T}_i$.
Since only the structure matters for neighbor checking, and not the actual values, this data structure can also be used for the neighbor checking step.

The pseudo-code for the vertex selection is given in Function~\ref{algo:grgvws} and the edges can be computed as described in Section~\ref{sec:allneigh}.

\begin{megaalgorithm}[ht]
 \KwData{Node $t_{k,p}$ (in $\mathcal{T}_i$), Current node $n_{k_i,p_i}$}
 \eIf{($\|p-p_i\|_2-\frac{\sqrt{d}}{2}2^{k_i} \geq \alpha 2^{k}$ AND doesNotContainPath($t_{k,p}$)}{
 	Remove all descendants of $t_{k,p}$\;
 	\If{$n_{k,p}$ has not been sampled yet}{
 		Sample $N_{\rm{samples}}$ in $H(n_{k,p})$\;
 		Estimate $\hat{V}(n_{k,p})$\;
 	}
 }{
 	\ForEach{$(m,q)$ index of children of $(k,p)$}{
 		\If{$t_{m,q}$ does not exist}{Create $t_{m,q}$ child of $t_{k,p}$\;}
 		{ \tt \small GetRGVerticesWithSampling}($t_{m,q}$,$n_{k_i,p_i}$)\;
 	}
 }

 \caption{\tt GetRGVerticesWithSampling()}
\label{algo:grgvws}
\end{megaalgorithm}

\section{Minimal Reduced Graph Construction}

Constructing $\mathcal{G}_i$ can be costly and only part of the information might be used at each iteration to solve for the shortest path.
In this work, it is assumed that the planning problem on $\mathcal{G}_i$ is solved using the A$^*$ algorithm although this is not restrictive.
In the A$^*$ algorithm, nodes are kept in a priority queue, called $OPEN$, ordered by $f$-values, where $f=g+h$ with $g$ the cost-to-go and $h$ an admissible heuristic to the goal. While $OPEN$ has elements, the first element is removed and for each of the neighbors, if they have not been closed yet, the $g$-value is updated, and it is added to the $OPEN$ priority queue. The algorithm stops when the first element of the $OPEN$ priority queue is the goal.

In the A$^*$ algorithm, knowing the neighbors of a node is only useful when that node is taken out of the $OPEN$ queue.
Similarly, the obstacle probability is only needed to calculate the {$g$-value} of a node.

By delaying those calculations until the necessary information is required, improvement in execution speed is expected. The following changes allow to save computations in the new algorithm:

\begin{itemize}
\item the {\tt ReducedGraph} function only computes the nodes of the reduced graph
\item during the $A^*$ algorithm, neighbors of a node are computed when the node is selected from the OPEN priority queue to be explored. If sampling is being used, sampling is only made the first time the $g$-value is calculated.
\end{itemize}

At the end, the algorithm will only have calculated the neighbors for the nodes in the $CLOSE$ list, and estimated the obstacle probability for nodes in $OPEN \cup CLOSE$.
In the worst case, the A$^*$ algorithm will explore every vertex and every edge, so all neighbors will be calculated and all nodes will be sampled similarly to the na\"{\i}ve case. But, in general, the number of neighbors calculated and the number of node sampled will be largely reduced compared to the na\"{\i}ve case.
This is confirmed by the numerical examples in the next section.

\section{Probabilistic bounds of MSPP-S}
As the estimate $\hat{V}(n_{k,p})$ is used instead of the actual obstacle probability $V(n_{k,p})$, a bad estimate could lead to missing solutions and losing completeness of the algorithm. In particular, if $\hat{V}(n_{k,p})$ overestimates $V(n_{k,p})$, the node $n_{k,p}$ might wrongly be evaluated as a $\varepsilon$-obstacle which would prevent the algorithm from finding any path passing through $n_{k,p}$, and potentially the only solution, thus breaking the completeness of the algorithm.

In this section, we derive an analytic worst case bound for the probability of failure of the MSPP-S algorithm.

We assume that there is an underlying grid and that each unit cell is either free space or obstacle.

\subsection{Definitions}
To deal with the probability of misevaluating the obstacle probability of a cell, we redefine the notion of obstacles with a threshold $\gamma$ and the event of wrongly evaluating a node as an obstacle.
\begin{defi}
Let $\varepsilon,\gamma>0$. A node $n_{k,p}$ is a $\varepsilon,\gamma$-obstacle if
\begin{equation}
\label{gobs}
\hat{V}(n_{k,p})\geq 1-2^{-d k}\varepsilon + \gamma.
\end{equation}
\end{defi}
\begin{defi}
Let $\varepsilon,\gamma>0$. $M(n_{k,p})$ is the event that the node $n_{k,p}$ is a $\varepsilon,\gamma$-obstacle and is not a $\varepsilon$-obstacle.
\end{defi}
\subsection{Bounds on $P \big(M(n_{k,p})\big)$}
\begin{prop}
Let $\varepsilon,\gamma>0$ and $n$ the number of sampled points in a node $n_{k,p}$. Then
\begin{equation}
P \big(M(n_{k,p})\big) \leq exp\left(-2\gamma^2 n\right).
\end{equation}
\end{prop}
\begin{proof}
Suppose $M(n_{k,p})$ is true. Then (\ref{obstacle_def}) and (\ref{gobs}) are verified, then from (\ref{gobs})-(\ref{obstacle_def}), we get  
\begin{equation}
\label{vhatminusv}
\hat{V}(n_{k,p})-V(n_{k,p}) \geq \gamma.
\end{equation}
Suppose that $n_{k,p}$ is composed of $N$ unit cells, including $N_o$ obstacles. Let $\mu=N_o/N$, then 
$V(n_{k,p}) = N_o/N =\mu$.
Uniformly sampling a random point in $n_{k,p}$ is similar to uniformly picking one of the $N$ unit cells, that is, an obstacle is picked with probability $\frac{N_o}{N} =\mu$. Let $x_i$ be the random variable associated with the $i^{th}$ obstacle test. $x_i$ takes value 1 for obstacles and 0 for free space, that is $x_i$ is a Bernouilli trial. Let $n_o=\sum_{i=1}^n x_i$. $n_o$ is the number of successes in $n$ Bernoulli trials, so $n_o$ follows a binomial distribution. Using (\ref{vhatminusv}), changing variables and using Hoeffding's inequality ((\ref{11}) to (\ref{12})), we get that
\begin{eqnarray}
&& P \big(M(n_{k,p})\big) \\
&\leq & P \left(\hat{V}(n_{k,p})-V(n_{k,p}) \geq \gamma \right)\\
 &=& P\left( \frac{n_o}{n} - \mu \geq \gamma \right) \\
\label{11} &=& P \big( n_o \geq n(\gamma+\mu) \big) \\
\label{12} &\leq& exp\left(-2\gamma^2 n \right).
\end{eqnarray}
\end{proof}
\begin{prop}
Let a node $n_{k,p}$.
If $k \geq k_{max} =\lceil \frac{1}{d}\log_2 \frac{\varepsilon}{\gamma} \rceil$, then $n_{k,p}$ cannot be a $\varepsilon,\gamma$-obstacle and $M(n_{k,p})$ never happens.
\end{prop}
\begin{proof}
\begin{eqnarray}
&& k \geq k_{max} =\left \lceil \frac{1}{d}\log_2 \frac{\varepsilon}{\gamma} \right\rceil \geq \frac{1}{d}\log_2 \frac{\varepsilon}{\gamma} \\
&\Rightarrow& 2^{dk} \geq \frac{\varepsilon}{\gamma}\\
&\Rightarrow& \gamma \geq \varepsilon 2^{-dk}\\
&\Rightarrow& 1-2^{-d k}\varepsilon + \gamma \geq 1.
\end{eqnarray}
\end{proof}
\begin{prop}
Let a node $n_{k,p}$.
If $k \leq k_{min} = \lfloor \frac{1}{d} \log_2 n  \rfloor$, then computing $V(n_{k,p})$ is less expensive than computing $\hat{V}(n_{k,p})$. So $\varepsilon$-obstacles can be used and $M(n_{k,p})$ never happens.
\end{prop}
\begin{proof}
The cost of computing $V(n_{k,p})$ is $2^{dk}$. The cost of computing $\hat{V}(n_{k,p})$ is $n$.
\begin{eqnarray}
&& k \leq k_{min} = \left\lfloor \frac{1}{d} \log_2 n  \right\rfloor \leq \frac{1}{d} \log_2 n \\
&\Rightarrow& 2^{dk} \leq n.
\end{eqnarray}
\end{proof}

\subsection{Probability of failure of MSPP-S}
We assume here the MSPP-S algorithm uses $\varepsilon,\gamma$-obstacles when $k>k_{min}$ and $\varepsilon$-obstacles otherwise.
We also assume that the algorithm evaluates $\hat{V}(n_{k,p})$ at most once for each node, it is evaluated when the value is needed and then the value is kept in memory.
\begin{prop}
Given $\varepsilon,\gamma,n>0$, an upperbound on the probability of failure of MSPP-S is given by
\begin{equation}
P(failure) \leq 1-\left(1-exp(-2\gamma^2 n)\right)^{nb_{occ}},
\end{equation}
where
\begin{equation}
nb_{occ} =\frac{2^{d(\ell-k_{min})}-2^{d(\ell-k_{max}+1)}}{2^d-1}.
\end{equation}
\end{prop}
\begin{proof}
$M(n_{k,p})$ can happen for every node such that $k_{min}< k < k_{max}$, that is, the maximum number of occurences $nb_{occ}$ of $M(n_{k,p})$ is the number of nodes verifying $k_{min}< k < k_{max}$,
\begin{eqnarray}
nb_{occ} &=& \sum_{k_{min}< k < k_{max}} 2^{d(\ell-k)} \\
&=& \frac{2^{d(\ell-k_{min})}-2^{d(\ell-k_{max}+1)}}{2^d-1}.
\end{eqnarray}
If $M(n_{k,p})$ never happens, the algorithm will never discard potential paths, and stay complete, that is,
\begin{eqnarray}
P(success)& \geq & P\big(M(n_{k,p})\text{ never happens}\big) \\
& = & P\big(\neg M(n_{k,p})\big)^{nb_{occ}}.
\end{eqnarray}
Hence,
\begin{equation}
P(failure) \leq 1-\left(1-exp(-2\gamma^2 n)\right)^{nb_{occ}}.
\end{equation}
\end{proof}

\subsection{Upperbound when multiple independent solutions exist}
Suppose there exists $Z$ independant solutions to the planning problem, that is, the search space can de partitionned in $Z$ regions in which there exists at least one solution. Suppose for simplificity that all regions have the same size, $1/Z^{th}$ of the original space.\\
The probability of failing in one region is then bounded by 
\begin{equation}
P(failure\text{ in one region}) \leq 1-\left(1-exp(-2\gamma^2 n)\right)^{nb_{occ}/Z}.
\end{equation}
The probability of the algorithm failing is smaller that the algorithm failing in each regions simultaneously since there might exist solutions crossing regions, so
\begin{equation}
P(failure) \leq \left( 1-\left(1-exp(-2\gamma^2 n)\right)^{nb_{occ}/Z}\right)^Z.
\end{equation}

Fig~\ref{fig:prob_bound} shows the upperbound on the probability of failure as a function of the number of sampled points $n$ for a given set of parameters of the algorithm. For small $n$, the probability is very close to 1 since we have a very poor precision in the estimate of the obstacle probabilities. As $n$ grows, the probabilty for every single cell gets better hence the probability of failure decreases. The grow of $n$ also increases the value of $k_{min}$, thus creating drops in the maximum number of occurences of $M\left( n_{k,p} \right)$ and in the probability of failure. As $k_{min}$ reaches $k_{max}-1$ the maximum number of occurences of $M\left( n_{k,p} \right)$ goes to 0 and the probability of failure of the algorithm then becomes 0.

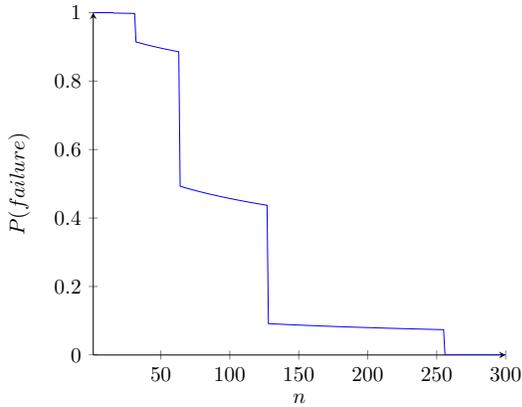
\begin{figure}
\centering
\begin{tikzpicture}[scale=0.8]
\begin{axis}[
    axis lines = left,
    xlabel = {$n$},
    ylabel = {$P(failure)$},
]
\addplot[
    color=blue,
    ]
    coordinates {
    (1,1)(2,1)(3,1)(4,1)(5,1)(6,1)(7,1)(8,1)(9,1)(10,1)(11,1)(12,1)(13,1)(14,1)(15,1)(16,0.99872)(17,0.99865)(18,0.99857)(19,0.9985)(20,0.99842)(21,0.99835)(22,0.99828)(23,0.9982)(24,0.99813)(25,0.99806)(26,0.99798)(27,0.99791)(28,0.99784)(29,0.99777)(30,0.99769)(31,0.99762)(32,0.91437)(33,0.91324)(34,0.91213)(35,0.91103)(36,0.90995)(37,0.90889)(38,0.90785)(39,0.90682)(40,0.90581)(41,0.90482)(42,0.90383)(43,0.90286)(44,0.90191)(45,0.90096)(46,0.90003)(47,0.89911)(48,0.89821)(49,0.89731)(50,0.89642)(51,0.89555)(52,0.89468)(53,0.89383)(54,0.89298)(55,0.89214)(56,0.89131)(57,0.89049)(58,0.88968)(59,0.88888)(60,0.88808)(61,0.8873)(62,0.88652)(63,0.88574)(64,0.49295)(65,0.49174)(66,0.49054)(67,0.48935)(68,0.48818)(69,0.48703)(70,0.48589)(71,0.48477)(72,0.48366)(73,0.48257)(74,0.48149)(75,0.48042)(76,0.47936)(77,0.47832)(78,0.47728)(79,0.47626)(80,0.47526)(81,0.47426)(82,0.47327)(83,0.4723)(84,0.47133)(85,0.47038)(86,0.46943)(87,0.4685)(88,0.46757)(89,0.46666)(90,0.46575)(91,0.46485)(92,0.46396)(93,0.46308)(94,0.46221)(95,0.46135)(96,0.4605)(97,0.45965)(98,0.45881)(99,0.45798)(100,0.45715)(101,0.45634)(102,0.45553)(103,0.45473)(104,0.45393)(105,0.45314)(106,0.45236)(107,0.45159)(108,0.45082)(109,0.45006)(110,0.4493)(111,0.44855)(112,0.44781)(113,0.44707)(114,0.44634)(115,0.44561)(116,0.44489)(117,0.44418)(118,0.44347)(119,0.44277)(120,0.44207)(121,0.44138)(122,0.44069)(123,0.44001)(124,0.43933)(125,0.43866)(126,0.43799)(127,0.43733)(128,0.091568)(129,0.091363)(130,0.09116)(131,0.090958)(132,0.090758)(133,0.09056)(134,0.090363)(135,0.090168)(136,0.089975)(137,0.089783)(138,0.089593)(139,0.089404)(140,0.089216)(141,0.08903)(142,0.088846)(143,0.088663)(144,0.088481)(145,0.0883)(146,0.088121)(147,0.087944)(148,0.087767)(149,0.087592)(150,0.087418)(151,0.087246)(152,0.087074)(153,0.086904)(154,0.086735)(155,0.086568)(156,0.086401)(157,0.086236)(158,0.086071)(159,0.085908)(160,0.085746)(161,0.085585)(162,0.085425)(163,0.085267)(164,0.085109)(165,0.084952)(166,0.084797)(167,0.084642)(168,0.084489)(169,0.084336)(170,0.084184)(171,0.084034)(172,0.083884)(173,0.083735)(174,0.083587)(175,0.083441)(176,0.083295)(177,0.08315)(178,0.083005)(179,0.082862)(180,0.08272)(181,0.082578)(182,0.082438)(183,0.082298)(184,0.082159)(185,0.08202)(186,0.081883)(187,0.081747)(188,0.081611)(189,0.081476)(190,0.081342)(191,0.081208)(192,0.081076)(193,0.080944)(194,0.080812)(195,0.080682)(196,0.080552)(197,0.080423)(198,0.080295)(199,0.080168)(200,0.080041)(201,0.079915)(202,0.079789)(203,0.079664)(204,0.07954)(205,0.079417)(206,0.079294)(207,0.079172)(208,0.079051)(209,0.07893)(210,0.078809)(211,0.07869)(212,0.078571)(213,0.078453)(214,0.078335)(215,0.078218)(216,0.078101)(217,0.077985)(218,0.07787)(219,0.077755)(220,0.077641)(221,0.077527)(222,0.077414)(223,0.077302)(224,0.07719)(225,0.077079)(226,0.076968)(227,0.076857)(228,0.076748)(229,0.076638)(230,0.07653)(231,0.076422)(232,0.076314)(233,0.076207)(234,0.0761)(235,0.075994)(236,0.075888)(237,0.075783)(238,0.075678)(239,0.075574)(240,0.075471)(241,0.075367)(242,0.075265)(243,0.075162)(244,0.07506)(245,0.074959)(246,0.074858)(247,0.074758)(248,0.074658)(249,0.074558)(250,0.074459)(251,0.07436)(252,0.074262)(253,0.074164)(254,0.074067)(255,0.07397)(256,0)(257,0)(258,0)(259,0)(260,0)(261,0)(262,0)(263,0)(264,0)(265,0)(266,0)(267,0)(268,0)(269,0)(270,0)(271,0)(272,0)(273,0)(274,0)(275,0)(276,0)(277,0)(278,0)(279,0)(280,0)(281,0)(282,0)(283,0)(284,0)(285,0)(286,0)(287,0)(288,0)(289,0)(290,0)(291,0)(292,0)(293,0)(294,0)(295,0)(296,0)(297,0)(298,0)(299,0)(300,0)
    };
 
\end{axis}
\end{tikzpicture}
\caption{Bound on the probability of failure with the parameters $\ell=5$, $d=1$, $\varepsilon=90\%$, $\gamma=0.35\%$ and $Z=2$.}
\label{fig:prob_bound}
\end{figure}

The upperbounds derived are very conservative in the sense that:
\begin{itemize}
\item in most cases, only part of the nodes $n_{k,p}$ with $k_{min}< k < k_{max}$ are evaluated;
\item if $M(n_{k,p})$ happens for a node that is not part of the solution, the solution will still be found by the algorithm;
\item if multiple solutions exist but are not independent, the probability is still largely reduced, but not exponentially;
\item in typical environement, large areas of free space exists, thus multiplying the number of possible solutions and largely reducing the probability of failure.
\end{itemize}  
In practise, failure of the MSPP-S algorithm has not been observed.

\section{Results}

\subsection{Comparison in Random Environemnts}

In this section, we compare the original MSPP algorithm against the proposed extensions and also against the A$^*$ run on a uniform grid.
MSPP-FN refers to the variant with the new neighbor test and MSPP-S refer to the variant using sampling (it also uses fast neighbors).
Obstacle maps were randomly generated and then used to solve path-planning problems via these four algorithms.
The problem was solved for dimensions ranging from 2 to 5 with a tree depth of 5, that is, for search spaces ranging from $2^{2 \times 5}=1024$ to $2^{5\times 5}\simeq 3\times 10^7$.

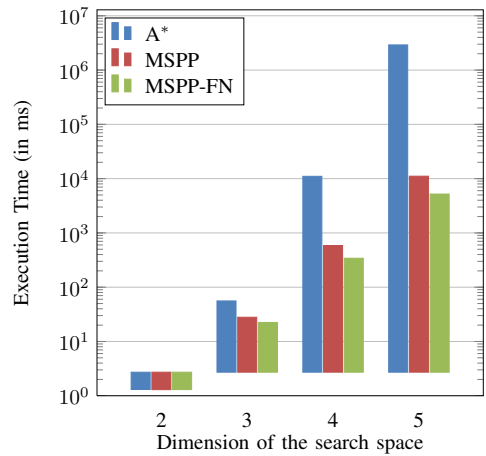
\begin{figure}[ht]
\centering
\input{result1.tex}
\caption{Comparison of the execution time of the A$^*$, MSPP and MSPP-FN algorithms. The results are shown in logarithmic scale.}
\label{comparison}
\end{figure}

Figure~\ref{comparison} shows the average execution time (in log scale) of the MSPP, the MSPP-FN and the A$^*$ algorithms on the randomly generated maps.
In this figure, the time to create the map is not taken into account in order to compare the pure performance of the planning algorithms, that is, it is just the time to find a path on an already existing map.

For the smaller search spaces, we see very few differences between all the algorithms, as expected.
As the dimension and the size of the search space grow however, the MSPP algorithm becomes much faster than the A$^*$, by more than two orders of magnitude in dimension 5. By the same token,
the MSPP-FN algorithm is even faster (by 50\%) over the baseline MSPP algorithm in dimension 5.

In Figure~\ref{comparison2}, the cost of creating the map is taken into account. This is done in order
to compare the results of using the MSPP-S algorithm.
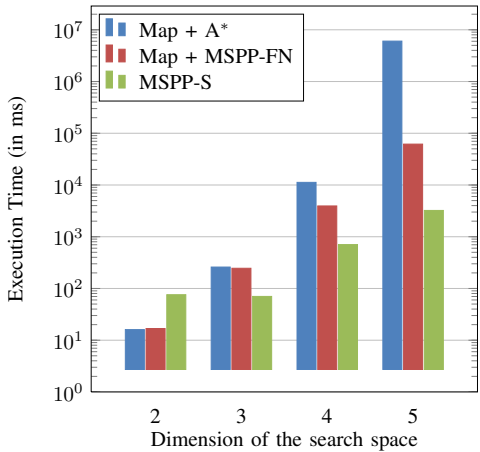
\begin{figure}[ht]
\centering
\input{result2.tex}
\caption{Comparison of the time to construct the map and run the A$^*$ or MSPP-FN algorithm versus the time to run the MSPP-S algorithm for which a map does not need to be computed. The results are shown in logarithmic scale.}
\label{comparison2}
\end{figure}
Three algorithms are compared here, the A$^*$ algorithm with the construction of the graph, the MSPP algorithm with the construction of the multi-scale map and the MSPP-S algorithm.
Similarly to the previous case, on a small search space, there is little or no improvement.
As the problem dimension increases, however, the improvement gets much better.
The MSPP-S algorithm is three orders of magnitude faster than creating a map and using the A$^*$ algorithm, and more than ten times faster than constructing a multi-scale map and using the original MSPP algorithm.

\subsection{Application to a Robot Arm}

\begin{figure}[ht]
\centering
\begin{subfigure}{0.45\linewidth}
\includegraphics[width=0.99\textwidth]{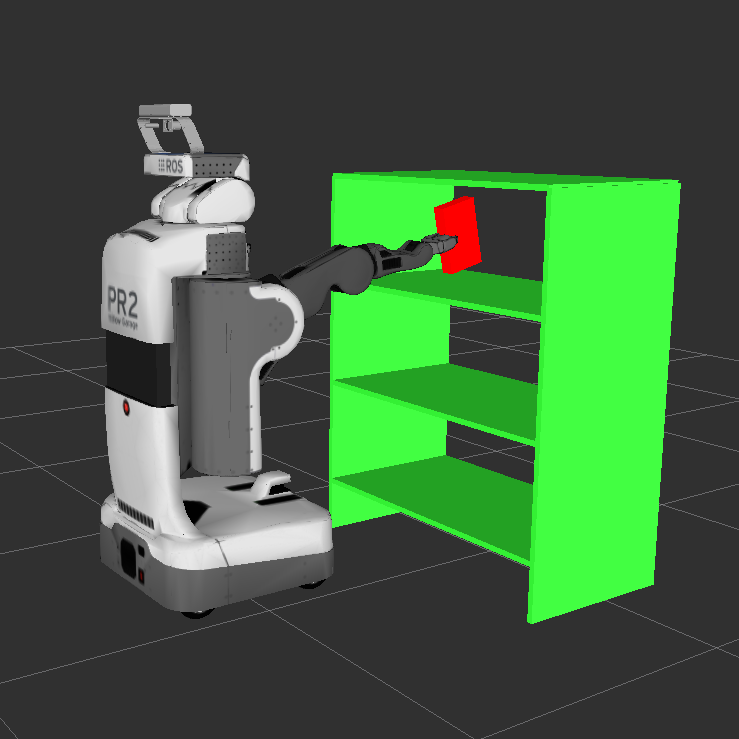}
\end{subfigure}%
\hspace*{2ex}
\begin{subfigure}{0.45\linewidth}
\includegraphics[width=0.99\textwidth]{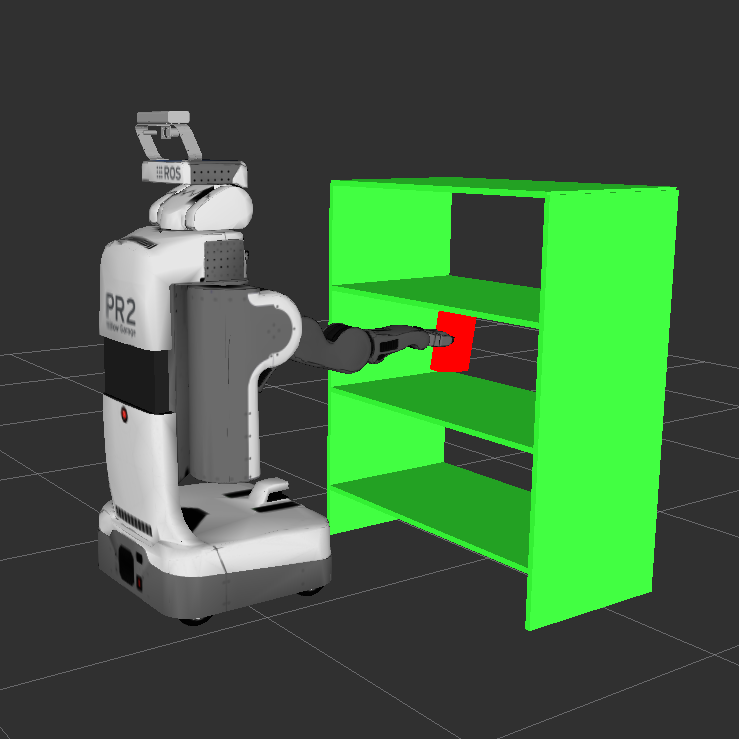}
\end{subfigure}
\caption{Initial and final pose of the planning problem for the PR2 arm.}
\label{pr2}
\end{figure}

The planning algorithm was used to plan a trajectory for an arm of the PR2 robot. Planning was done in the configuration space using 4 joints of the arm. Figure~\ref{pr2} shows the initial configuration and the desired final configuration; the robot needs to move a book from the top shelf to the second shelf. The depth of the tree was set to 5, creating a search space of size $2^{4\times 5}\simeq 3\times 10^7$.

The path-planning problem was solved three times to compare the variants of the algorithm: First, the multi-scale map was built by exploring the entire search-space and the MSPP algorithm was used to find the solution.
All algorithms were performed on the same desktop computer running Ubuntu Linux.
Building the map was the most time-consuming process. 
It takes on average 4 minutes and 52 seconds and solving the path-planning problem takes on average 47 seconds.
Using the MSPP-FN algorithm on the same map, the problem was solved in 4 seconds on average.
%

\section{Conclusions}

In this paper, we have introduced several modifications and extensions to the original MSPP algorithm, first presented in \cite{hauer2015multi}, to increase its computational efficiency.
The resulting multi-scale path-planning algorithms, called MSPP-FN and MSPP-S offer several non-trivial improvements over the previous MSPP algorithm.
First, the complexity of each iteration of the algorithm is reduced by changing the manner by which the adjacency relationships in the reduced graph are computed at each iteration.
Second, the range of applications of the algorithm has been widened,
by allowing the use of an obstacle predicate rather than accurate prior knowledge of a full information multi-scale map.
This extension results in much fewer requirements in terms of memory allocation and theoretical bounds on the probability of failure were derived.
Third, reordering the operations performed by the algorithm allows one to minimize computations by avoiding
information that is not needed during execution.
We have compared the original MSPP algorithm to the proposed MSPP-FN and MSPP-S algorithms and found runtime improvements by over 50\%.
Both algorithms outperform A$^*$ by more than two orders of magnitude.

\bibliographystyle{abbrv}
\bibliography{flo}

\end{document}

%% file: same_size.tex
\newcommand{\drawChildren}[2]{\draw[treenodes] #1 rectangle ($#1+(#2,#2)$);
\draw[treenodes] #1 rectangle ($#1+(-#2,#2)$);
\draw[treenodes] #1 rectangle ($#1+(-#2,-#2)$);
\draw[treenodes] #1 rectangle ($#1+(#2,-#2)$);}

\begin{tikzpicture}[scale=\textwidth/28cm]
\tikzstyle{treenodes}=[black,thick]

\draw[black,thick] (-3,-1) rectangle (-1,1);
\draw[black,thick] (-1,-3) rectangle (1,-1);
\draw[black,thick] (-1,1) rectangle (1,3);
\draw[black,thick] (1,-1) rectangle (3,1);

\draw[black,thick] (-1,-1) rectangle (1,1);

\node[draw,circle,inner sep=1pt,fill] (nkp) at (0,0){};
\node[anchor=north west] at (nkp.south east) {$p$};

\node[draw,circle,inner sep=1pt,fill] (n1) at (2,0){};
\node[anchor=north] at (n1.south) {$p_{\rm{nhb},1}$};
\node[draw,circle,inner sep=1pt,fill] (n2) at (0,2){};
\node[anchor=south] at (n2.north) {$p_{\rm{nhb},2}$};
\node[draw,circle,inner sep=1pt,fill] (n3) at (-2,0){};
\node[anchor=north] at (n3.south) {$p_{\rm{nhb},3}$};
\node[draw,circle,inner sep=1pt,fill] (n4) at (0,-2){};
\node[anchor=north] at (n4.south) {$p_{\rm{nhb},4}$};

\path[draw,dashed,->,red] (nkp) -- (n1);
\path[draw,dashed,->,red] (nkp) -- (n2);
\path[draw,dashed,->,red] (nkp) -- (n3);
\path[draw,dashed,->,red] (nkp) -- (n4);

\node[] (d1) at (-2,-3){};
\node[] (d2) at (-3,-2){};
\node[] (o) at (-3,-3){};
\node[anchor=west] at (d1.east) {$d_1$};
\node[anchor=south] at (d2.north) {$d_2$};
\path[draw,->] (o.center) -- (d1);
\path[draw,->] (o.center) -- (d2);

\end{tikzpicture}

%% file: larger.tex
\newcommand{\drawChildren}[2]{\draw[treenodes] #1 rectangle ($#1+(#2,#2)$);
\draw[treenodes] #1 rectangle ($#1+(-#2,#2)$);
\draw[treenodes] #1 rectangle ($#1+(-#2,-#2)$);
\draw[treenodes] #1 rectangle ($#1+(#2,-#2)$);}

\begin{tikzpicture}[scale=\textwidth/50cm]
\tikzstyle{treenodes}=[black,thick]

\draw[black,thick] (-3,-1) rectangle (-1,1);
\draw[black,thick] (-1,-3) rectangle (1,-1);
\draw[black,thick] (-1,1) rectangle (1,3);
\draw[black,thick,dashed] (1,-1) rectangle (3,1);
\draw[black,thick,dashed] (1,-1) rectangle (5,3);
\draw[white,thick] (1,-5) rectangle (9,3);
\draw[black,thick,dashed] (1,-5) rectangle (9,3);
\draw[white,thick] (-7,-5) rectangle (9,11);
\draw[black,thick,dashed] (-7,-5) rectangle (9,11);

\draw[black,thick] (-1,-1) rectangle (1,1);

\node[draw,circle,inner sep=1pt,fill] (nkp) at (0,0){};
\node[anchor=east] at (nkp.west) {\small{$p$}};

\node[draw,circle,inner sep=1pt,fill] (n1) at (2,0){};
\node[anchor=west] at (n1.east) {\small{$p_{\rm{nhb},1}$}};

\path[draw,dashed,->,red] (nkp) -- (n1);

\node[] (d1) at (-4,-4){$d_1$};
\node[] (d2) at (-6,-2){$d_2$};
\node[] (o) at (-6,-4){};
\path[draw,->] (o.center) -- (d1);
\path[draw,->] (o.center) -- (d2);

\end{tikzpicture}

%% file: smaller.tex
\newcommand{\drawChildren}[2]{\draw[treenodes] #1 rectangle ($#1+(#2,#2)$);
\draw[treenodes] #1 rectangle ($#1+(-#2,#2)$);
\draw[treenodes] #1 rectangle ($#1+(-#2,-#2)$);
\draw[treenodes] #1 rectangle ($#1+(#2,-#2)$);}

\begin{tikzpicture}[scale=\textwidth/10cm]
\tikzstyle{treenodes}=[black,thick,dashed]

\drawChildren{(1.25,0.25)}{0.25};
\drawChildren{(1.5,0.5)}{0.5};
\drawChildren{(1.5,-0.5)}{0.5};
\draw[white,thick] (1,0) rectangle (2,1);
\draw[black,thick,dashed] (1,0) rectangle (2,1);
\draw[white,thick] (1,-1) rectangle (2,0);
\draw[black,thick,dashed] (1,-1) rectangle (2,0);
\draw[white,thick] (1,-1) rectangle (3,1);
\draw[black,thick,dashed] (1,-1) rectangle (3,1);

\draw[black,thick] (-1,-1) rectangle (1,1);

\draw[blue,thick] (1,0.5) rectangle (1.5,1);
\draw[blue,thick] (1,0.25) rectangle (1.25,0.5);
\draw[blue,thick] (1,0) rectangle (1.25,0.25);
\draw[blue,thick] (1,-0.5) rectangle (1.5,0);
\draw[blue,thick] (1,-1) rectangle (1.5,-0.5);

\node[draw,circle,inner sep=1pt,fill] (nkp) at (0,0){};
\node[anchor=east] at (nkp.west) {$p$};

\node[draw,circle,inner sep=1pt,fill] (n1) at (2,0){};
\node[anchor=west] at (n1.east) {$p_{\rm{nhb},1}$};

\path[draw,dashed,->,red] (nkp) -- (n1);

\node[] (d1) at (-0.4,-0.8){$d_1$};
\node[] (d2) at (-0.8,-0.4){$d_2$};
\node[] (o) at (-0.8,-0.8){};
\path[draw,->] (o.center) -- (d1);
\path[draw,->] (o.center) -- (d2);

\end{tikzpicture}

%% file: result1.tex
\begin{tikzpicture}[scale=0.8]
    \begin{axis}[
        width  = 0.45*\textwidth,
        height = 8cm,
        major x tick style = transparent,
        ybar=2*\pgflinewidth,
        bar width=9pt,
        ymajorgrids = true,
        ylabel = {Execution Time (in ms)},
        xlabel = {Dimension of the search space},
        ymode = log,
        symbolic x coords={2,3,4,5},
        xtick = data,
        scaled y ticks = false,
        enlarge x limits=0.25,
        ymin=1,
        legend cell align=left,
        legend style={
                at={(0.02,0.98)},
                anchor=north west,
                column sep=1ex
        }
    ]
        \addplot[style={bblue,fill=bblue,mark=none}]
            coordinates {(2, 1.3052) (3,55.8653) (4,11039.4) (5,2907080)};

        \addplot[style={rred,fill=rred,mark=none}]
             coordinates {(2,1.2961) (3,27.9657) (4,586.148) (5,11101.3)};

        \addplot[style={ggreen,fill=ggreen,mark=none}]
             coordinates {(2,1.2991) (3,22.437) (4,341.891) (5,5200.94)};

        \legend{A$^*$,MSPP,MSPP-FN}
    \end{axis}
\end{tikzpicture}

%% file: result2.tex
\begin{tikzpicture}[scale=0.8]
    \begin{axis}[
        width  = 0.45*\textwidth,
        height = 8cm,
        major x tick style = transparent,
        ybar=2*\pgflinewidth,
        bar width=9pt,
        ymajorgrids = true,
        ylabel = {Execution Time (in ms)},
        xlabel = {Dimension of the search space},
        ymode = log,
        symbolic x coords={2,3,4,5},
        xtick = data,
        scaled y ticks = false,
        enlarge x limits=0.25,
        ymin=1,
        legend cell align=left,
        legend style={
                at={(0.02,0.98)},
                anchor=north west,
                column sep=1ex
        }
    ]
        \addplot[style={bblue,fill=bblue,mark=none}]
            coordinates {(2, 16.0344) (3,259.53) (4,11250) (5,6043900)};

        \addplot[style={rred,fill=rred,mark=none}]
             coordinates {(2,16.84) (3,246.0137) (4,3937.6) (5,61736)};

        \addplot[style={ggreen,fill=ggreen,mark=none}]
             coordinates {(2,75.8) (3,70.2365) (4,705.3750) (5,3224.9)};

        \legend{Map + A$^*$,Map + MSPP-FN ,MSPP-S}
    \end{axis}
\end{tikzpicture}